\documentclass[11pt]{article}
\usepackage{geometry}                
\usepackage{graphicx}
\usepackage{amssymb,amsmath,amsfonts,amsthm,bm}
\usepackage{authblk}
\usepackage{natbib}
\usepackage{ifthen}
\usepackage[colorinlistoftodos, textwidth=2cm, obeyFinal]{todonotes}
\usepackage{mathtools}
\usepackage{framed}
\usepackage{enumitem}
\usepackage[linesnumbered,lined,ruled]{algorithm2e}

\newtheorem{theorem}{Theorem}
\newtheorem{lemma}{Lemma}

\newtheorem{proposition}{Proposition}
\theoremstyle{definition}
\newtheorem{definition}{Definition}

\newcommand{\R}{\mathbb{R}}

\newcommand{\Prob}[2][n]{\mathbf{P}\SwitchBracketsizeLeft{#1}\LeftBracketSize\{#2\SwitchBracketsizeRight{#1}\RightBracketSize\}}

\newcommand{\vconv}{\mathrm{\mathop{\textbf{vconv}}}}
\newcommand{\conv}{\mathop{\textbf{conv}}}

\allowdisplaybreaks

\newcommand{\abs}[2][n]{\SwitchBracketsizeLeft{#1}\LeftBracketSize\lvert#2\SwitchBracketsizeRight{#1}\RightBracketSize\rvert}

\newcommand{\set}[3][a]{\SwitchBracketsizeLeft{#1}\LeftBracketSize\{#2 : #3\SwitchBracketsizeRight{#1}\RightBracketSize\}}

\newcommand{\NextScriptStyle}[1]{{\scriptstyle{#1}}}
\newcommand{\NextScriptScriptStyle}[1]{{\scriptscriptstyle{#1}}}
\newcommand{\NextTextStyle}[1]{{\textstyle{#1}}}
\newcommand{\NextDisplayStyle}[1]{{\displaystyle{#1}}}
\newcommand{\SwitchBracketsizeLeft}[1]{
  \ifthenelse{\equal{#1}{b}\OR\equal{#1}{big}}{\let\LeftBracketSize=\bigl}{
    \ifthenelse{\equal{#1}{B}\OR\equal{#1}{Big}}{\let\LeftBracketSize=\Bigl}{
      \ifthenelse{\equal{#1}{g}\OR\equal{#1}{bigg}}{\let\LeftBracketSize=\biggl}{
    \ifthenelse{\equal{#1}{G}\OR\equal{#1}{Bigg}}{\let\LeftBracketSize=\Biggl}{
      \ifthenelse{\equal{#1}{s}\OR\equal{#1}{small}}{\let\LeftBracketSize=\NextScriptStyle}{
        \ifthenelse{\equal{#1}{ss}}{\let\LeftBracketSize=\NextScriptScriptStyle}{
          \ifthenelse{\equal{#1}{t}\OR\equal{#1}{text}}{\let\LeftBracketSize=\NextTextStyle}{
        \ifthenelse{\equal{#1}{d}\OR\equal{#1}{display}}{\let\LeftBracketSize=\NextDisplayStyle}{
          \ifthenelse{\equal{#1}{a}\OR\equal{#1}{auto}}{\let\LeftBracketSize=\left}{
            \let\LeftBracketSize=\relax}}}}}}}}}}
\newcommand{\SwitchBracketsizeRight}[1]{
  \ifthenelse{\equal{#1}{b}\OR\equal{#1}{big}}{\let\RightBracketSize=\bigr}{
    \ifthenelse{\equal{#1}{B}\OR\equal{#1}{Big}}{\let\RightBracketSize=\Bigr}{
      \ifthenelse{\equal{#1}{g}\OR\equal{#1}{bigg}}{\let\RightBracketSize=\biggr}{
    \ifthenelse{\equal{#1}{G}\OR\equal{#1}{Bigg}}{\let\RightBracketSize=\Biggr}{
      \ifthenelse{\equal{#1}{s}\OR\equal{#1}{small}}{\let\RightBracketSize=\NextScriptStyle}{
        \ifthenelse{\equal{#1}{ss}}{\let\RightBracketSize=\NextScriptScriptStyle}{
          \ifthenelse{\equal{#1}{t}\OR\equal{#1}{text}}{\let\RightBracketSize=\NextTextStyle}{
        \ifthenelse{\equal{#1}{d}\OR\equal{#1}{display}}{\let\RightBracketSize=\NextDisplayStyle}{
          \ifthenelse{\equal{#1}{a}\OR\equal{#1}{auto}}{\let\RightBracketSize=\right}{
            \let\RightBracketSize=\relax}}}}}}}}}}

\title{A linear time algorithm for multiscale quantile simulation}
\author[1,$\star$]{Chengcheng~Huang}
\author[2]{Housen~Li}
\author[1]{Lizhi~Cheng}
\author[1]{Wei~Peng}

\affil[1]{College of Liberal Arts and Sciences, National University of Defense Technology, 410073 Changsha, China}
\affil[2]{Institute for Mathematical Stochastics, University of G\"ottingen, Goldschmidtstrasse 7, 37077 G\"ottingen, Germany}

\affil[$\star$]{Correspondence: {\tt huangchengcheng12@nudt.edu.cn}}

\date{}                                           

\begin{document}


\maketitle

\begin{abstract}
Change-point problems have appeared in a great many applications for example cancer genetics, econometrics and climate change. Modern multiscale type segmentation methods are considered to be a statistically efficient approach for multiple change-point detection, which minimize the number of change-points under a multiscale side-constraint. The constraint threshold plays a critical role in balancing the data-fit and model complexity. However, the computation time of such a threshold is quadratic in terms of sample size $n$, making it impractical for large scale problems. In this paper we proposed an $\mathcal{O}(n)$ algorithm by utilizing the hidden quasiconvexity structure of the problem. It applies to all regression models in exponential family with arbitrary convex scale penalties. Simulations verify its computational efficiency and accuracy. An implementation is provided in R-package ``linearQ'' on CRAN.  
\end{abstract}



\section{Introduction}\label{sec:intro}

In this paper, we assume that observations $Y=(Y_1,\dots,Y_n)$ are independent from the regression model 
\begin{align}
\label{model}
Y_i \thicksim F_{\vartheta(i/n)},\qquad i=0,\dots,n-1
\end{align}
where $\{F_{\theta}\}_{\theta \in \Theta}$ is a one-dimensional exponential family distribution with densities $f_\theta$.  The parametric function $\vartheta:[0,1) \rightarrow \Theta \subseteq \mathbb{R}$ is a right-continuous piecewise constant function. The model \eqref{model} includes the Gaussian mean regression as a special case, that is,
\begin{equation}\label{GaussMean}
Y_i \thicksim {\vartheta(i/n)} + \sigma\varepsilon_i,\qquad i=0,\dots,n-1, 
\end{equation}
where $\sigma > 0$ and $\varepsilon_i\stackrel{\mathrm{i.i.d.}}{\sim} \mathcal{N}(0,1)$ the standard Gaussian.

The (multiple) change-point problem amounts to estimating the number and locations of change-points and the value of function $\vartheta$ on each segment. 
The study of change-point detection problems has a long and rich history in the statistical literatures \citep{Carlstein1994,Cs2011,DavidSiegmund2013}, and has experienced a revival in recent years, mainly due to modern large scale applications, for example in bioinformatics, predicting transmembrane helix locations \citep{Lio2000}, detecting changes in the DNA copy number \citep{Olshen2004,Venkatraman2007};  in climate, analyzing undocumented change-points in climate data \citep{Reeves2007};  and in economics and finance, identifying change-points in financial volatility \citep{Spo09}. 

Among the vast literature of change-point problems,  we consider the so-called multiscale change-point segmentation methods \citep[see e.g.][]{FriMunSie14,LiMuSi16}, which are statistically well-understood and meanwhile practically well-performed, see also \citep{DavHoeKra12,Hotz2013,LiGuMu17}. These multiscale segmentation methods minimize the number of the change-points subjected to a side-constraint that  multiscale statistics $T_n$ does not exceed a specified threshold $q$ (see Section \ref{ss:mcs} for a formal definition). The threshold $q$, as a balancing parameter between the data-fit and model complexity, is often chosen as the quantile of $T_n$ under null distribution (e.g.,~$\vartheta \equiv 0$).  
Unfortunately, the computation of such a quantile involves the evaluation of $T_n$, which has quadratic computational time in terms of sample size, and has to be repeated sufficiently many times to guarantee a proper estimation accuracy. This makes the multiscale segmentation methods impractical for large scale applications (e.g, for $n \ge 100,000$). To overcome this computation bottleneck, we proposed, in this paper, a fast algorithm with linear computational complexity for the evaluation of $T_n$, see Figure \ref{fig:demo} for an illustration. 

\begin{figure}[!h]
\centering
 \includegraphics[width=4.5in]{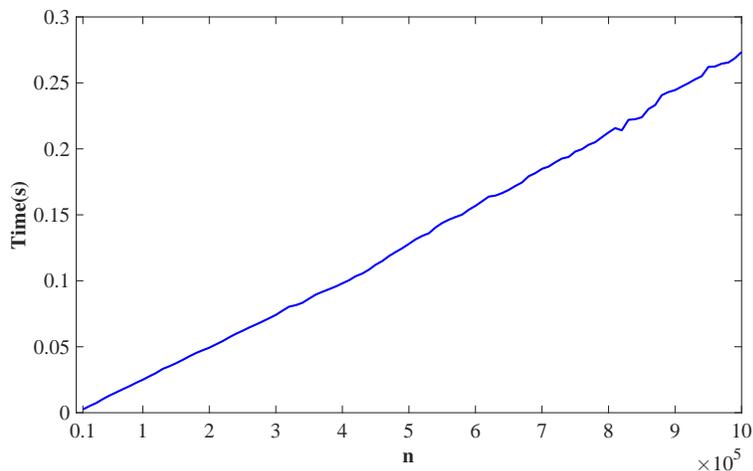}
\caption{Average computation time of $T_n$ over $100$ independent repetitions. \label{fig:demo}} 
\end{figure}

The rest of the paper is organized as follows. In Section~\ref{s:bg}, we introduce the multiscale change-point segmentation methods, and some basic concepts from algorithmic geometry.  In Section~\ref{s:main} we propose a linear algorithm for the evaluation of  $T_n$ and give its complexity analysis.  The performance of the proposed algorithm is examined by simulations in Section~\ref{s:sim}. Section~\ref{s:cld} concludes this paper.

\section{Background}\label{s:bg}
\subsection{Multiscale change-point segmentation}\label{ss:mcs}
We start with a brief introduction of the multiscale change-point segmentation methods for the change-point problem in \eqref{model}. Recall that the underlying truth $\vartheta$ is right-continuous and piecewise constant, i.e.,
\begin{align*}
\vartheta(t) = \sum \limits_{m=0}^M \theta_m \bm{\mathrm{1}}_{[\tau_m,\tau_{m+1})}(t).
\end{align*}
where $0 = \tau_0 < \dots < \tau_{M+1}=1$ denote the locations of change-points, and $\theta_m \in \mathbb{R}$ the function value on the $m$th segment with $\theta_m \neq \theta_{m+1}$.   Let $\mathcal{S}$ denote the space of all right continuous step functions, and, for every $\vartheta$ in $\mathcal{S}$, let $\mathcal{V}(\vartheta)$ denote the set of change-points and $\#\mathcal{V}(\vartheta)$  the number of change-points. 

The \emph{multiscale change-point segmentation estimator}  $\hat{\vartheta}$ is  the solution of  the optimization problem \citep{FriMunSie14, LiMuSi16,LiGuMu17}:
\begin{align}
\label{optim}
\inf_{\vartheta \in \mathcal{S}} \#\mathcal{V}(\vartheta) \qquad\qquad \text{subject to}~~T_n(Y,\vartheta)\leq q.
\end{align}
where $q$ is a user-specified threshold, and $T_n(Y,\vartheta)$ a \emph{multiscale statistic}. By $\mathcal{L}$  we denote the collection of all subintervals of $[0,1)$. The multiscale statistic $T_n(Y,\vartheta)$ is defined as the maximum of penalized local likelihood ratio statistic on every interval $I \in \mathcal{L}$ where $\vartheta = \theta_I$ is constant, that is, 
\begin{align}
\label{Tn}
T_n(Y,\vartheta) = \max_{\substack{{I \in \mathcal{L}} \\ {\vartheta(t)=\theta_I, ~t\in I}}} {T_I(Y,\theta)} - p_I.
\end{align}
Here the penalty terms $p_I$ play a role as scale (i.e., length of $I$) calibration, which aim to put different scales on equal baseline especially for small intervals \citep[see][]{Dumbgen2001,FriMunSie14}. The local likelihood ratio statistic $T_I(Y,\theta)$ is a testing statistic on the hypothesis $H_0: \theta = \theta_0$ versus the alternative $H_1: \theta \neq \theta_0$ with $\theta \equiv \vartheta(t)$ on interval $I$, more precisely, 
\begin{equation}
\label{localstatistics}
T_I(Y,\theta_0) = \sqrt{ \log\biggl( \frac{\sup_{\theta \in \Theta} \prod_{k/n \in I} f_\theta(Y_k)}{\prod_{k/n \in I} f_{\theta_0}(Y_k)}\biggr)}. 
\end{equation}
Note that the specific form $\sqrt{\log(\cdot)}$ in \eqref{localstatistics} is crucial  if one wants to use a simple, additive penalty term that yields statistical optimality, as in \eqref{Tn}, see \cite{RiWa13}.

The user-specific threshold $q \in \mathbb{R}$ in \eqref{optim} controls the probability of overestimating and underestimating the number of change-points. From asymptotic analysis, it is sufficient to choose a universal threshold $q \asymp \sqrt{\log n}$, see \cite{LiGuMu17}. In practice, it is recommended to select $q\coloneqq q_n(\alpha)$ the $1-\alpha$ quantile of null asymptotic distribution of $T_n(Y,\vartheta)$ with certain significance level $\alpha \in [0,1)$, which allows for an immediate statistical interpretation
\[
\Prob{\#\mathcal{V}(\hat\vartheta)\le \#\mathcal{V}(\vartheta)} \ge 1-\alpha
\]
see \cite{FriMunSie14}. Given such choices of $q$, the solution to problem \eqref{optim} exists but may be non-unique, in which case one is free to choose the solution, such as the constrained maximum likelihood estimator  \citep{FriMunSie14}.  
Note that the value of $q_n(\alpha)$ can be estimated via Monte Carlo simulations, because the distribution of $T_n(Y,\vartheta)$ or its asymptotic distribution \citep[Theorem 2.1]{FriMunSie14} is independent of the unknown truth $\vartheta$. 

\subsection{Constrained Minkowski sum}\label{ss:cMs}
We now restrict ourselves to the Euclidean space $\R^2$. The Minkowski sum, a fundamental concept in algorithmic geometry,  is defined as $ P \oplus Q = \{ p+q \,|\, p\in P, q\in Q \}$ for $P, Q \subseteq \mathbb{R}^2$. As in \cite{BEH09}, we define the constrained Minkowski sum as
\begin{align*}
 (P + Q)^+= \{ x \in P \oplus Q | \, x_1 > 0\}   \qquad \text {with } x_1 \text{ the first coordinate of point } x \in \R^2
 \end{align*}
By $\conv(P)$ we denote the convex hull of $P$, and by $\vconv(P)$ the set of vertices of $\conv(P)$. \cite{BerHof06}, \cite{BEH07,BEH09} have shown that $\vconv(P+Q)^+$ can be computed in $\mathcal{O}(\abs{P} +\abs{Q})$ time, if $P$ and $Q$ are sorted with respect to the first coordinate. More precisely, a set $R$  can be computed such that $\vconv(P + Q)^+ \subseteq R  \subseteq (P + Q)^+$ and $\abs{R} \le \min\{2\cdot \abs{P}+\abs{Q},\,\abs{P}+2\cdot \abs{Q} \} -2$.

For general (not necessarily ordered) $P$ and $Q$, the computation of $\vconv(P+Q)$ requires $\mathcal{O}\bigl((\abs{P}+\abs{Q})\log(\abs{P}+\abs{Q})\bigr)$ runtime, where the additional log factor is due to sorting algorithms. See \citep[e.g.][]{Fukuda2004,Weibel2007} for the computation of Minkowski sum in $\mathbb{R}^d$ with $d \ge 2$.  

\subsection{Quasiconvexity}
We recall some basic results of quasiconvexity, a useful generalization of convexity, see e.g., \citep[Section 4 in Chapter 3]{BoVa04} for further details and the proofs.  
\begin{definition}
\label{quasiconvex}
Let $\mathcal{D} \subseteq \mathbb{R}^d$ be a nonempty convex set. A function $f: \mathcal{D} \rightarrow \mathbb{R}$ is called \emph{quasiconvex} if its sublevel set $\mathcal{D}_\alpha :=\{x\in \mathcal{D}\,| \, f(x)\leq \alpha\}$ is convex for every $\alpha \in \mathbb{R}$.
\end{definition}
Note that convex functions are clearly quasiconvex, and that many properties of convex functions carry over to quasiconvex functions.
\begin{proposition}
\label{pp:max}
let $\mathcal{D} \subseteq \mathbb{R}^d$ be a nonempty convex set. A function $f:\mathcal{D}\rightarrow \mathbb{R}$ is quasiconvex if and only if for any $s_1,s_2 \in \mathcal{D}$ and any $\lambda \in [0,1]$ it holds that
$$
f(\lambda \cdot  s_1 + (1-\lambda)\cdot s_2 ) \leq \max{\{f(s_1),f(s_2)\}}.
$$
\end{proposition}

\section{An $\mathcal{O}(n)$ method for quantile simulation }\label{s:main}
In this section, we first consider the computation of quantiles $q_n(\alpha)$ for multiscale change-point segmentation methods.  We will show that the evaluation of $T_n(Y,\vartheta)$ is equivalent to finding the maximal value of a quasiconvex function over a constrained Minkowski sum. 

\subsection{Fast quantile simulation}\label{ss:fastQ}
We start with the Gaussian mean regression model \eqref{GaussMean}, and the penalty term $p_I$ given in \cite{FriMunSie14}. Note that in model \eqref{GaussMean} the distribution of $T_n(Y,\vartheta)$ is independent of $\vartheta$. Thus, it is sufficient to consider
\begin{align}
\label{normal}
T_n \coloneqq T_n(Y,\vartheta\equiv 0) = \max_{1 \le i \le j \le n}\frac{1}{\sigma\sqrt{j-i+1}}\abs[B]{\sum \limits_{k=i}^j Y_k} -  \sqrt{2\log(\frac{en}{j-i+1})}. 
\end{align}
The direct evaluation of \eqref{normal} leads to $\mathcal{O}(n^3)$ runtime. As the summation can be viewed as convolution, the evaluation of \eqref{normal} can be speeded up by utilizing fast Fourier transforms, resulting in  $\mathcal{O}(n^2 \log n)$ runtime (which is implemented in CRAN R-package ``stepR''), see e.g.~\cite{Hotz2013}.  A further speedup is possible by means of cumulative sum transformation $\mathrm{cs_m}:=\sum_{k=1}^mY_k$, which reduces a summation over $\{i, \ldots, j\}$ to a single subtraction. This leads to an algorithm of $\mathcal{O}(n^2)$ complexity (which is implemented in CRAN R-package ``FDRSeg''), see also \cite{Allison2003}. In what follows, we will present a fast algorithm for evaluating \eqref{normal} in a linear runtime, i.e., $\mathcal{O}(n)$. 

For $1\le i \le j \le n$, we define $s_{i,j}\coloneqq \sum_{k=i}^j Y_k $, and $\ell_{i,j}\coloneqq  j-i+1$. The evaluation of $T_n(Y,\vartheta\equiv 0)$ in \eqref{normal} can be written as an optimization of a bivariate function over finite collection of points, more precisely,
\begin{equation}\label{eq:2dOpt}
T_n= \max_{1 \le i \le j \le n} h(\ell_{i,j},s_{i,j}) \quad \text{with }h(x_1,x_2)\coloneqq \frac{\abs{x_2}}{\sigma\sqrt{x_1}}-\sqrt{2\log \frac{en}{x_1}}.
\end{equation}
 
\begin{proposition}
The bivariate function $h$ in \eqref{eq:2dOpt} is quasiconvex over $(0,n]\times\R$. 
\end{proposition}

\begin{proof} 
By Definition \ref{quasiconvex}, it is sufficient to show that sublevel set 
$$
\mathcal{D}_{\alpha} = \set{(x_1,x_2)}{\abs{x_2} \leq \sigma(\alpha +\sqrt{2\log\frac{en}{x_1}})\sqrt{x_1}, \text{ and } 0< x_1 \le n}
$$ 
is convex for all $\alpha \in \mathbb{R}$. Define $g(x_1) \coloneqq \bigl(\alpha + \sqrt{2\log (en/x_1)}\bigr)\sqrt{x_1}$ for $x_1 > 0$.  Notice that it is trivial when $g(x_1) < 0$ because sublevel set $\mathcal{D}_{\alpha}$ is empty. If $\mathcal{D}_{\alpha}$  is not empty, it follows that $\bigl(\alpha + \sqrt{2\log (en/x_1)}\bigr) \geq 0$. Noting that $x_1 \le n$ implies $\log{({en}/{x_1})} \geq 1$, we have  
$$
g''(x_1) = -\frac{1}{4}x_1^{-{3}/{2}}\biggl(\alpha + \sqrt{2\log{\frac{en}{x_1}}}\biggr) - \biggl(2x_1\log{\frac{en}{x_1}}\biggr)^{-{3}/{2}} < 0 .
$$
Thus, $g(\cdot)$ is concave, and it follows that $\mathcal{D}_{\alpha }$ is convex for all $\alpha$. 
\end{proof}

By Proposition~\ref{pp:max} we have that the maximal value of $f$ in \eqref{eq:2dOpt} over $\{(s_{i,j}, \ell_{i,j})\}_{i,j}$ is attained at the vertices of the convex hull of $\{(s_{i,j}, \ell_{i,j})\}_{i,j}$. To be precise, we define $P \coloneqq \set[n]{p_i}{p_i = (i, \sum_{j=1}^i Y_j),\, i = 1, \ldots, n}$ and $Q \coloneqq\set[n]{q_i}{q_i = (i-n, -\sum_{j=1}^{n-i}Y_j),\, i = 1,\ldots, n}$ with the convention that $\sum_{i = 1}^0 Y_i = 0$. Note that $(\ell_{i,j},s_{i,j}) = p_i + q_{n-j+1}$. It follows that
\[
T_n = \max_{x \in (P\oplus Q)^+} h(x) = \max_{x \in \mathrm{conv}(P\oplus Q)^+}h(x) = \max_{x \in \mathrm{vconv}(P\oplus Q)^+}h(x). 
\]
Moreover, it is known that there is a linear algorithm for finding $\vconv(P\oplus Q)^+$ (see Section~\ref{ss:cMs}). Based on it, we can derive a linear algorithm for the evaluation of $T_n$, the details of which is given in Algorithm~\ref{agg1}. 

\begin{algorithm}
\label{agg1}
  \SetAlgoLined
  \KwIn{Observations $Y_1,\ldots,Y_n$.}
  \KwOut{The value of  $T_n$ in \eqref{normal}.}
  \textbf{Initialization}: Define $P \coloneqq \{p_i\}_{i=1}^n$ with $p_i \equiv (i, \sum_{j=1}^iY_j)$, and 
   $Q\coloneqq\{q_i\}_{i = 1}^n$ with $q_i \equiv (i-n,-\sum_{j=1}^{n-i}Y_j)$;
Set $R,K_0,\bar{K}_0,D_0$ as the empty set in $\R^2$\;
   Apply the incremental Graham scan algorithm to $P$ (from $p_n$ to $p_1$)\;
   \For{$i = 1,\dots,n$ }{
   $p_i^u$ $\leftarrow$ the neighbor points of $p_i$ on  $\vconv(\{p_i, \ldots, p_n\}) \cap (\R\times \R_+)$\;
   $p_i^l$ $\leftarrow$ the neighbor points of $p_i$ on  $\vconv(\{p_i, \ldots, p_n\}) \cap (\R\times\R_-)$\;
   }
   Append points to point-set $R$ recursively\;
    \For{$i =1,\ldots, n$}{
    Compute $\vconv\{q_{n-i+1},\dots, q_{n}\}$ via the incremental Graham scan algorithm (from $q_{n}$ to $q_{n-i+1}$)\;
    \For{$q_j \in\mathbf{vconv}\{q_{n-i+1},\dots,q_{n}\}$}{
	 \If{$((0,0),p_i-p_i^u,q_j-q_{j+1})$ is counterclockwise}{
		$K_i\leftarrow$ $K_{i-1} \bigcup \{q_j\}$ \quad\qquad\qquad \# $q_j$ belongs to $\vconv(P\oplus Q)^+ \cap (\R\times \R_+)$\\
			}
	\ElseIf{$((0,0),p_i-p_i^l,q_j-q_{j+1})$ is clockwise}{
	$K_i\leftarrow$ $K_{i-1} \bigcup \{q_j\}$ \quad\qquad\qquad \# $q_j$ belongs to $\vconv(P\oplus Q)^+ \cap (\R\times \R_-)$\\
	}
    }
    $R \leftarrow  R \,\bigcup \, (\{p_i\}\oplus K_i)$\;
    $D_i \leftarrow$  $ \{q_{n-i+1},\dots,q_{n}\} \setminus \vconv\{q_{n-i+1},\dots,q_{n}\}$\;
    $\bar{K_i} \leftarrow \{q_{i_2}^*,\dots, q_{i_{\mu-1}}^*\}$ (if denote $K_{i}\equiv \{q_{i_1}^*,q_{i_2}^*,\dots,q_{i_{\mu-1}}^*,q_{i_\mu}^* \}$)\;
    Update $Q \leftarrow  Q \setminus \{D_i \cup \bar{K}_i\}$\;
    $i\leftarrow i+1$\;
    }
 Evaluate the value of $f$ in \eqref{eq:2dOpt} over $R$ and find the maximal value $T_n$.
 \caption{Evaluation of $T_n$ for the Gaussian mean regression model.}
 \end{algorithm}
 
In Algorithm \ref{agg1}, the incremental Graham scan algorithm \citep{Graham1972} is employed in first step to compute the convex hull of $P$ in $\mathcal{O}(n)$ runtime on line 2. For each point $p_i$, we consider $\conv\{q_{n-i+1}, \ldots ,q_n\}$ in order to satisfy the constraint $(p+q)_{x_1}> 0$~(line~9).  Among such points, we compute a set $K_i:=\{q_{i_1}^*,\dots,q_{i_\mu}^*\}$ that contains the vertices involving $p_i$ in $\vconv(P\oplus Q)^+$ ~(line 11-16). After recording $(p_i \oplus K_i)$  to $R$ (line 18), we delete  $\bar{K}_{i} \coloneqq \{q_{i_2}^*,\dots, q_{i_{\mu-1}}^*\}$ and $D_{i} \coloneqq\{q_{n-i+1},\dots, q_n\} \setminus \vconv(\{q_{n - i+1},\dots,q_n\})$ from $Q$ (line 21), because  there is no point in $\conv(P\oplus Q)^+$ of the form $p_j+q$ for $j > i $ and $q\in \bar{K}_i \bigcup D_i$, see \cite{BEH09} for a proof. Then the algorithm proceeds recursively; each time we update $R, K_i,\bar{K}_i$ and $ D_i$. In the end, the set $R$, being a subset of $ (P\oplus Q)^+$,  contains $\vconv(P\oplus Q)^+$. The maximal value $T_n$ can be obtained on $R$.

As $T_n$ in \eqref{normal} is independent of $\sigma$, we can always assume $\sigma = 1$. Given realization $\{Y_1, \ldots, Y_n\}\stackrel{\mathrm{i.i.d.}}{\sim} \mathcal{N}(0,1)$, we compute $T_n$ via Algorithm \ref{agg1}. The quantile of $T_n$ is computed via $r$ repetitions of such a procedure. Thus, the quantile of $T_n$ can be computed in $\mathcal{O}(nr)$ runtime. This is significantly faster than the best existing algorithm, which is of $\mathcal{O}(n^2r)$ runtime. 
In general, larger $r$ leads to more precise estimation of the quantile. In practice, we find that the estimate is quite stable for $r\ge 5,000$ (see Figure \ref{fig:qnt}), and thus suggest $r = 5,000$ as the default choice. 

\begin{figure}[!h] 
 \centering
 \includegraphics[width=4.5in]{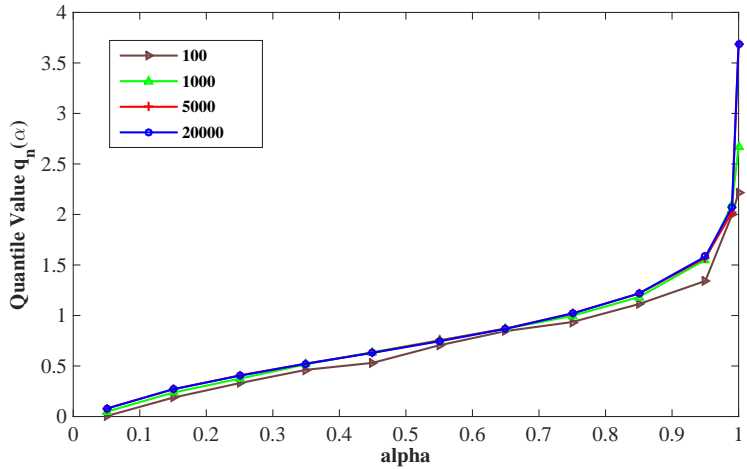}
 \caption{Empirical quantile function of $T_n$ in \eqref{normal} for different number of repetitions $r$.\label{fig:qnt}}
\end{figure}

\subsection{Quasiconvexity of exponential family regression}\label{ss:exp}

By Proposition~\ref{pp:max}, a larger class of quasiconvex objective functions $h(\ell_{i,j},s_{i,j})$ attain the maximum over $(P + Q)^+$ on $\vconv(P + Q)^+$.  In exponential family regression model, the multiscale statistic \eqref{Tn} is independent of $\vartheta$. Similar to Gaussian mean model \eqref{eq:2dOpt}, it can be mapped into a bivariate function $h:(0,n]\times\R \rightarrow \R$.  Next we will discuss the computation of multiscale statistics for the exponential family regression model \eqref{model} and extend Algorithm~\ref{agg1} to a general form.
 
We assume in model (\ref{model}), the independent observations $Y_i$ ($i = 0,\dots,n-1)$ come from a exponential family with standard density
$$
f_{\theta}(x) = \text{exp}\{\theta x - \psi(\theta)\},
$$
where the cumulant function $\psi(\theta)$ is strictly convex on $\Theta$. Then the maximum scanning statistics $T_I$ in \eqref{localstatistics} through exponential family regression can be simplified as:
\begin{equation}
\label{scanning}
T_I(Y,\theta_0) = \sup_{\theta \in \Theta}\sum \limits_{k=i}^j \{ \theta Y_k -\psi(\theta) \} - \sum \limits_{k=i}^j \{ \theta_0Y_k -\psi(\theta_0) \}.
\end{equation}
Let $\ell_{i,j}$ and $s_{i,j}$ as defined before, then the evaluation of $T_I$ in \eqref{scanning} can be written as a supremum over $\vartheta$ of a bivariate function:
\begin{equation}
\label{TIopt}
T_I(Y,\theta_0) = \sup_{\theta \in \Theta}h_{\theta}(\ell_{i,j},s_{i,j})\quad \text{with}~ h_{\theta}(\ell,s)\coloneqq(\theta-\theta_0)s-\ell(\psi(\theta)-\psi(\theta_0)).
\end{equation}

\begin{lemma}
\label{lem1}
The bivariate function $\sup_{\theta \in \Theta}h_{\theta}(\cdot,\cdot)$ in \eqref{TIopt} is convex over $ (0,n] \times \R$.
\end{lemma}
\begin{proof}
Since $h_\theta(\ell,s)$ is linear about $\ell, s$, it follows that $\forall (\ell_1,s_1),(\ell_2,s_2) \in (0,n]\times \R, \lambda \in [0,1]$,
\begin{align*}
\sup_{\theta \in \Theta}h_\theta((1-\lambda)\ell_1+\lambda \ell_2, (1-\lambda)s_1+\lambda s_2)&=\sup_{\theta \in \Theta}(\lambda h_\theta(\ell_1,s_1)+(1-\lambda)h_\theta(\ell_2,s_2))\\\nonumber
&\le \lambda \sup_{\theta \in \Theta}h_\theta(\ell_1,s_1) + (1-\lambda) \sup_{\theta \in \Theta}h_\theta(\ell_2,s_2).
\end{align*}
By the definition of convex function, $ \sup_{\theta \in \Theta}h_\theta$ is convex on $(0,n] \times \R$ .
\end{proof}
 
The multiscale statistic $T_n$ in \eqref{Tn} is made up of the scanning statistic $T_I$ and a penalty term $p_I$. The penalty function $p_I$ working as a scale calibration only depends on interval length $\ell$. 
So multiscale statistic $T_n$ can be written as $\sup_{\theta \in \Theta}h_{\theta}(\ell_{i,j},s_{i,j})$ in \eqref{TIopt} added by a penalty function:
\begin{equation}
\label{TnIopt}
T_n(Y,\theta_0) = \sup_{\theta \in \Theta}h_{\theta}(\ell_{i,j},s_{i,j})-p_I(\ell_{i,j}).
\end{equation}
By Lemma \ref{lem1}, the bivariate function $\sup_{\theta \in \Theta}h_{\theta}(\cdot,\cdot)$ is convex and it keeps convex if it is substracted by a concave penalty $p_I$. According to Proposition \ref{pp:max} the maximum of \eqref{TnIopt} over $\{(\ell_{i,j},s_{i,j})\}_{i,j}$ can be attained on the vertices of the convex hull of $\{(\ell_{i,j},s_{i,j})\}_{i,j}$. Thus, the optimization
of multiscale statistic $T_n$ can also be solved by Algorithm \ref{agg1}. We state this result in Theorem \ref{theorem1}.

\begin{theorem}
\label{theorem1}
The multiscale statistic $T_n$ for exponential family regression model with concave penalty terms can be evaluated in a linear runtime. 
\end{theorem}

In summary, the proposed algorithm is a general method for simulating multiscale statistic $T_n$  from exponential family regression model with convex penalization. It speeds up the existing algorithms to a linear runtime. Meanwhile, the memory space mainly used for storing points is bounded by the number of vertices in $\vconv(P+Q)^+$, i.e., $\mathcal{O}(\abs{P}+\abs{Q}) = \mathcal{O}(n)$.

\section{Simulation study}\label{s:sim}
This section examines the empirical performance of the proposed Algorithm \ref{agg1}. We provide the implementation of the proposed method in R package ``linearQ'', available from CRAN. 

\begin{figure}[!h]
 \centering
 \includegraphics[width=4.5in]{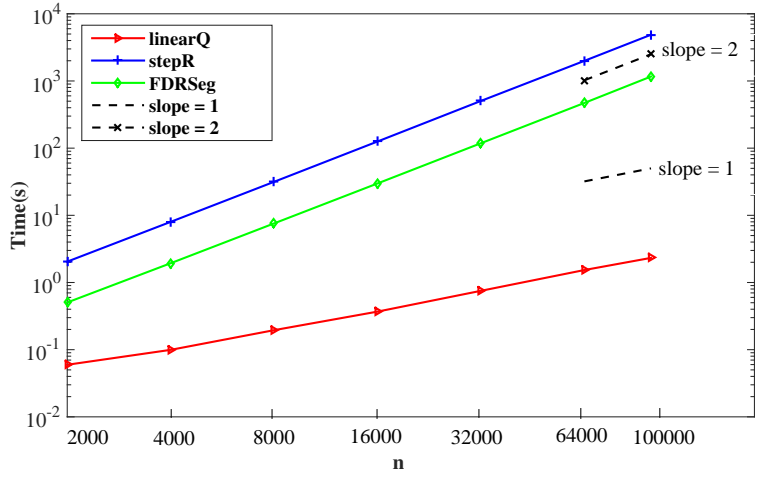}
 \caption{Gaussian mean regression: Average computation time of $T_n$ via various methods over 100 repetitions (both coordinates are in logarithmic scale). \label{fig3}}
 \end{figure}
 
We start with the Gaussian mean regression in \eqref{GaussMean}, and compare the proposed method with other existing methods. To this end, we consider the Fourier transform based algorithm, implemented in  CRAN R package ``stepR" \citep{FriMunSie14}, and  the cumulative sum based algorithm, implemented in CRAN R package ``FDRSeg" \citep{LiMuSi16}, see Section \ref{ss:fastQ}. The simulation data is generated as i.i.d.~realizations of standard normal random variables, for different sample sizes ranging from $2\times 10^3$ to $10^5$. For a given sample size, we repeat $r$ times, which is set to $100$. The average computation time for the evaluation of $T_n$ for different methods is reported in Figure~\ref{fig3}. It shows that the proposed method is significantly faster than the other two, achieving one order speed-up, with its computation complexity $\mathcal{O}(n)$.

 \begin{figure}[!h]
 \centering
 \includegraphics[width=4.5in]{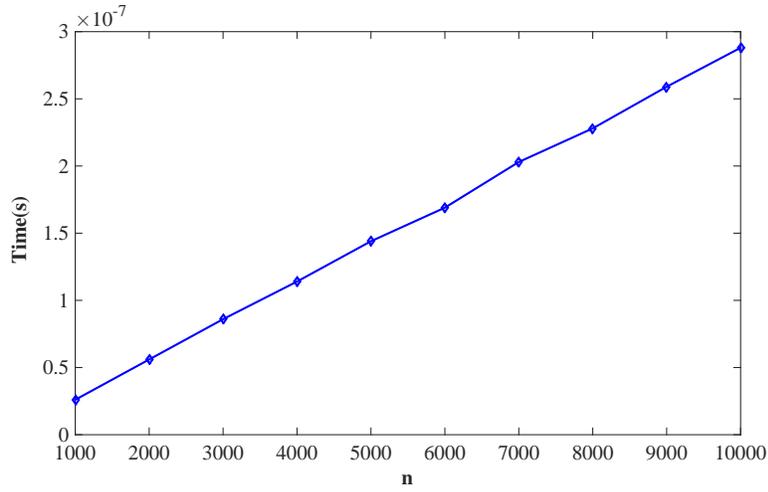}
 \caption{Poisson regression: Average computation time of $T_n$ via the proposed method over 100 repetitions. \label{fig4}}
 \end{figure}
 
In addition, the proposed method applies to every distribution in exponential family provided that the penalty term is convex, see Section \ref{ss:exp}. As a demonstration, we consider the Poisson case, i.e., $F_\theta$ in \eqref{model} is the Poisson distribution with mean $\theta$. Figure \ref{fig4} illustrates the computation time of evaluating $T_n$ without scale penalty, with the data sizes from $10^3$ to $10^4$ and repetition $r = 100$, via the proposed method. Again the empirical performance supports our theoretical complexity analysis that the computation time is linear in terms of sample size $n$.

\section{Conclusion}\label{s:cld}
The multiscale change-point segmentation methods are recognized as the-state-of-the-art in change-point inference, and have been playing an important role in various applications. In this paper, we propose a fast algorithm for the computation of the only tuning parameter of such multiscale change-point segmentation methods. The proposed method has a linear computation complexity and a linear memory complexity, in terms of the sample size, in sharp contrast to the existing methods with at least quadratic computation complexity. The crucial idea behind is to transform the original problem into the maximization of a quasiconvex function over a constrained Minkowski sum. The theoretical complexity is well supported by the empirical performance. Extension to general models beyond exponential family is a possible line of future research.

\bibliographystyle{apalike}
\bibliography{MultiscaleSegRef}

\end{document}